\documentclass[a4paper,12pt]{article}

\usepackage{multirow}
\usepackage{color}
\usepackage{amssymb, amsmath}
\usepackage{hyperref}
\usepackage{tikz}

\newcommand\EFFACE[1]{}

\newcommand\OR[1]{\overrightarrow{#1}}
\def\OG{\OR{G}}

\newtheorem{theorem}{Theorem}
\newtheorem{proposition}[theorem]{Proposition}

\newtheorem{lemma}[theorem]{Lemma}
\newtheorem{corollary}[theorem]{Corollary}

\newtheorem{observation}[theorem]{Observation}

\newenvironment{proof}{
\par
\noindent {\bf Proof.}\rm}{\mbox{}\hfill$\square$\par\vskip 3mm}

\newcommand\chiP{\chi'}
\newcommand\DP{D'}
\newcommand\chiD{D_{\chi}}
\newcommand\chiDP{D_{\chi'}}
\newcommand\OchiD{OD_{\chi}}
\newcommand\OchiDP{OD_{\chi'}}
\newcommand\OD{OD}

\newcommand\Dmin{{\OD}^{-}}
\newcommand\Dmax{{\OD}^{+}}
\newcommand\chiDmin{{\OD}_{\chi}^{-}}
\newcommand\chiDmax{{\OD}_{\chi}^{+}}
\newcommand\DminP{{{\OD}'}^{-}}
\newcommand\DmaxP{{{\OD}'}^{+}}
\newcommand\chiDminP{{\OD}_{\chi'}^{-}}
\newcommand\chiDmaxP{{\OD}_{\chi'}^{+}}

\def\Aut{{\rm Aut}}

\def\p2{\frac{p}{2}}  

\makeatletter
\let\@fnsymbol\@arabic
\makeatother

\begin{document}


\title{Distinguishing numbers and distinguishing indices\\ of oriented graphs}

\author{Kahina MESLEM~\thanks{DGRSDT, LaROMaD Laboratory, University of Science and Technology Houari Boumediene, Algiers, Algeria.}
\and \'Eric SOPENA~\thanks{Univ. Bordeaux, CNRS, Bordeaux INP, LaBRI, UMR5800, F-33400 Talence, France.}~$^,$\footnote{Corresponding author. Eric.Sopena@labri.fr}
}

\maketitle

\abstract{
A distinguishing $r$-vertex-labelling (resp. $r$-edge-labelling) of an undirected graph $G$ is a mapping $\lambda$ from the set of vertices (resp. the set of edges)
of $G$ to the set of labels $\{1,\dots,r\}$ such that 
no non-trivial automorphism of $G$ preserves all the vertex (resp. edge) labels.
The distinguishing number $D(G)$ and the distinguishing index $D'(G)$ of $G$ are then 
the smallest $r$ for which $G$ admits a distinguishing $r$-vertex-labelling or $r$-edge-labelling, respectively.
The distinguishing chromatic number $\chiD(G)$ and the distinguishing chromatic index $\chiDP(G)$
are defined similarly, with the additional requirement that the corresponding labelling must be 
a proper colouring.

These notions readily extend to oriented graphs, by considering arcs instead of edges.
In this paper, we study the four corresponding parameters for oriented graphs whose underlying graph is a path, a cycle,
a complete graph or a bipartite complete graph.
In each case, we determine their minimum and maximum value, taken over all possible
orientations of the corresponding underlying graph, except for 
the minimum values for unbalanced
complete bipartite graphs $K_{m,n}$ with 
$m=2$, $3$ or $4$ and $n>3$, $6$ or $13$, respectively,
or $m\ge 5$ and  
$n > 2^m-\left\lceil\frac{m}{2}\right\rceil$, for which we only
provide upper bounds.
}

\medskip


\medskip


\section{Introduction}

All graphs considered in this paper are simple.
For a graph $G$, we denote by $V(G)$ its set of vertices, and by $E(G)$ its set of edges.

An \emph{$r$-vertex-labelling} of a graph $G$ is a mapping $\lambda$ from $V(G)$
to the set of labels $\{1,\dots,r\}$.
An \emph{$r$-vertex-colouring} of $G$ is a proper $r$-vertex-labelling of $G$, that is, 
an $r$-vertex-labelling $\lambda$ such that
$\lambda(u)\neq\lambda(v)$ for every edge $uv$ of $G$.
The \emph{chromatic number} $\chi(G)$ of $G$ is then the smallest number of labels (called colours in that case) needed
for a vertex-colouring of $G$.
Similarly, an \emph{$r$-edge-labelling} of $G$ is a mapping $\lambda'$ from $E(G)$
to the set of labels $\{1,\dots,r\}$, and
an \emph{$r$-edge-colouring} of $G$ is a proper $r$-edge-labelling of $G$, that is,
an $r$-edge-labelling $\lambda'$ such that
$\lambda'(e)\neq\lambda'(e')$ for every two adjacent edges $e$ and $e'$ (that is, such that $e$ and $e'$
have one vertex in common).
The \emph{chromatic index} $\chi'(G)$ of $G$ is then the smallest number of labels (or colours) needed
for an edge-colouring of $G$.

An \emph{automorphism} $\phi$ of a graph~$G$ is an edge-preserving mapping
from $V(G)$ to $V(G)$, that is, such that $uv\in E(G)$ implies $\phi(u)\phi(v)\in E(G)$.
For a given vertex or edge-labelling of~$G$,
an automorphism $\phi$ of~$G$ is \emph{$\lambda$-preserving} if $\lambda(\phi(u))=\lambda(u)$ for
every vertex $u$ of $G$, or $\lambda(\phi(uv))=\lambda(uv)$ for
every edge $uv$ of $G$, respectively.
A vertex or edge-labelling  $\lambda$ of $G$ is \emph{distinguishing} if the only
$\lambda$-preserving automorphism of $G$ is the identity, that is, the labelling $\lambda$
breaks all the symmetries of $G$. Such a distinguishing vertex or edge-labelling is \emph{optimal}
if $G$ does not admit any vertex or edge-labelling using less colours.

The \emph{distinguishing number}, \emph{distinguishing chromatic number}, 
\emph{distinguishing index} and \emph{distinguishing chromatic index} of a graph $G$,
denoted by $D(G)$, $\chiD(G)$, $D'(G)$ and $\chiDP(G)$, respectively,
are then defined as the smallest $r$ for which $G$ admits
a distinguishing $r$-vertex-labelling, a distinguishing $r$-vertex-colouring,
a distinguishing $r$-edge-labelling or a distinguishing $r$-edge-colouring,
respectively.
Distinguishing numbers and distinguishing chromatic numbers 
have been introduced by Albertson and Collins in~\cite{AC96} and
Collins and Trenk in~\cite{CT06}, respectively,
while distinguishing indices and distinguishing chromatic indices have been introduced by
Kalinowski and Pil\'sniak in~\cite{KP15}
(these two parameters are often denoted $\chi_D$ and $\chi'_D$
instead of $\chiD$ and $\chiDP$, respectively).

\begin{table}
\small{
\begin{center}
  \begin{tabular}{|c|c|c|c|c||c|c||c|c|}
\hline
 & Graph $G$ & $\chi(G)$ & $D(G)$ & $\chiD(G)$ & $\Dmin(G)$ & $\Dmax(G)$ &  $\chiDmin(G)$ & $\chiDmax(G)$  \\
\hline
\hline
\multirow{2}{*}{1.} & $P_{2n}$, & \multirow{2}{*}{2} & \multirow{2}{*}{2} & \multirow{2}{*}{2} & \multirow{2}{*}{1 (Th.~\ref{th:simple-chiDmin})} & \multirow{2}{*}{1 (Th.~\ref{th:paths})}  & \multirow{2}{*}{2 (Th.~\ref{th:simple-chiDmin})} & \multirow{2}{*}{2 (Th.~\ref{th:paths})}  \\
   & $n\ge 1$ & & & & & & &  \\
\hline
\multirow{2}{*}{2.} & $P_{2n+1}$, & \multirow{2}{*}{2} & \multirow{2}{*}{2} & \multirow{2}{*}{3} & \multirow{2}{*}{1 (Th.~\ref{th:simple-chiDmin})} & \multirow{2}{*}{2 (Th.~\ref{th:paths})}  & \multirow{2}{*}{2 (Th.~\ref{th:simple-chiDmin})} & \multirow{2}{*}{3 (Th.~\ref{th:paths})}  \\
   & $n\ge 1$ & & & & & & &  \\
\hline
\hline
3. & $C_{4}$ & 2 & 3 & 4 & 1 (Th.~\ref{th:simple-chiDmin}) & 2 (Th.~\ref{th:cycles}) & 2 (Th.~\ref{th:simple-chiDmin}) & 4 (Th.~\ref{th:cycles}) \\
\hline
4. & $C_{5}$ & 3 & 3 & 3 & 1 (Th.~\ref{th:simple-chiDmin}) & 2 (Th.~\ref{th:cycles}) & 3 (Th.~\ref{th:simple-chiDmin}) & 3 (Th.~\ref{th:cycles}) \\
\hline
5. & $C_{6}$ & 2 & 2 & 4 & 1 (Th.~\ref{th:simple-chiDmin}) & 2 (Th.~\ref{th:cycles}) & 2 (Th.~\ref{th:simple-chiDmin}) & 3 (Th.~\ref{th:cycles}) \\
\hline
\multirow{2}{*}{6.} & $C_{2n}$, & \multirow{2}{*}{2} & \multirow{2}{*}{2} & \multirow{2}{*}{3} & \multirow{2}{*}{1 (Th.~\ref{th:simple-chiDmin})} & \multirow{2}{*}{2 (Th.~\ref{th:cycles})} & \multirow{2}{*}{2 (Th.~\ref{th:simple-chiDmin})} & \multirow{2}{*}{3 (Th.~\ref{th:cycles})} \\
   & $n\ge 4$ & & & & & & &  \\
\hline
\multirow{2}{*}{7.} & $C_{2n+1}$, & \multirow{2}{*}{3} & \multirow{2}{*}{2} & \multirow{2}{*}{3} & \multirow{2}{*}{1 (Th.~\ref{th:simple-chiDmin})} & \multirow{2}{*}{2 (Th.~\ref{th:cycles})} & \multirow{2}{*}{3 (Th.~\ref{th:simple-chiDmin})} & \multirow{2}{*}{3 (Th.~\ref{th:cycles})} \\
   & $n\ge 3$ & & & & & & &  \\
\hline
\hline
\multirow{2}{*}{8.} & $K_{n}$, & \multirow{2}{*}{$n$} & \multirow{2}{*}{$n$} & \multirow{2}{*}{$n$} & \multirow{2}{*}{1 (Th.~\ref{th:simple-chiDmin})} & \multirow{2}{*}{2 (Th.~\ref{th:complete-graphs})} & \multirow{2}{*}{$n$ (Th.~\ref{th:simple-chiDmin})} & \multirow{2}{*}{$n$ (Th.~\ref{th:complete-graphs})} \\
   & $n\ge 3$ & & & & & & &  \\
\hline
\hline
\multirow{2}{*}{9.} & $K_{1,n}$, & \multirow{2}{*}{2} & \multirow{2}{*}{$n$} & \multirow{2}{*}{$n+1$} & \multirow{2}{*}{$\left\lceil\frac{n}{2}\right\rceil$ (Th.~\ref{th:K_1n})} & \multirow{2}{*}{$n$  (Th.~\ref{th:K_1n})} & 1+$\left\lceil\frac{n}{2}\right\rceil$ & {$n+1$} \\
   & $n\ge 2$ & & & & & & (Th.~\ref{th:K_1n}) & (Th.~\ref{th:K_1n}) \\
\hline
\multirow{2}{*}{10.} & $K_{n,n}$,  & \multirow{2}{*}{2} & \multirow{2}{*}{$n+1$} & \multirow{2}{*}{$2n$} & \multirow{2}{*}{1 (Th.~\ref{th:K_nn})} & \multirow{2}{*}{$n$ (Th.~\ref{th:K_nn})} & \multirow{2}{*}{2 (Th.~\ref{th:K_nn})} & \multirow{2}{*}{$2n$ (Th.~\ref{th:K_nn})} \\
   & $n\ge 2$ & & & & & & &  \\
\hline
\multirow{2}{*}{11.} & $K_{m,n}$,  & \multirow{2}{*}{2} & \multirow{2}{*}{$n$} & \multirow{2}{*}{$m+n$} & Th.~\ref{th:complete-bipartite-graphs} and  &  \multirow{2}{*}{$n$ (Th.~\ref{th:complete-bipartite-graphs})} & Th.~\ref{th:complete-bipartite-graphs} and & $m+n$  \\
   & $n > m\ge 2$ & & & & Cor.~\ref{cor:final-complete-bipartite}&  & Cor.~\ref{cor:final-complete-bipartite} &  (Th.~\ref{th:complete-bipartite-graphs}) \\
\hline
  \end{tabular}
\end{center}
}
\caption{Table of results for $\Dmin(G)$, $\Dmax(G)$, $\chiDmin(G)$ and $\chiDmax(G)$.}
\label{table:global-vertex}
\end{table}

A graph $G$ is \emph{rigid} (or \emph{asymmetric}) if the only automorphism of $G$ is the identity. Therefore,
$D(G)=1$ if and only if $G$ is rigid and, similarly, $D'(G)=1$ if and only if $G$ is rigid.
Moreover,
for every such graph $G$, $\chiD(G)=\chi(G)$ and $\chiDP(G)=\chi'(G)$.
Note here that being rigid is not a necessary condition for any of these two equalities to hold
(consider the path of order~2 or the path of order~3, respectively).

\medskip

Our aim in this paper is to study 
the distinguishing number, distinguishing chromatic number, 
distinguishing index and distinguishing chromatic index
of several classes of oriented graphs.
By oriented graphs, we mean here antisymmetric digraphs, that is, digraphs with no
directed cycle of length at most~2, or, equivalently, digraphs obtained
from undirected graphs by giving to each of their edges one of its two possible
orientations.
All the notions of vertex-labelling, vertex-colouring, edge-labelling, edge-colouring,
automorphism, distinguishing labelling, distinguishing number, distinguishing chromatic number, 
distinguishing index and distinguishing chromatic index, readily extend to oriented graphs by
simply considering arcs instead of edges. 
For each undirected graph~$G$ with $m$ edges, and each distinguishing parameter,
we will study both the minimum and maximum possible value of the parameter, taken
over all of the $2^m$ possible orientations of~$G$.

\medskip

\begin{table}
\small{
\begin{center}
  \begin{tabular}{|c|c|c|c|c||c|c||c|c|}
\hline
 & Graph $G$ & $\chi'(G)$ & $D'(G)$ & $\chiDP(G)$ & $\DminP(G)$ & $\DmaxP(G)$ &  $\chiDminP(G)$ & $\chiDmaxP(G)$  \\
\hline
\hline
\multirow{2}{*}{1.} & $P_{2n}$, & \multirow{2}{*}{2} & \multirow{2}{*}{2} & \multirow{2}{*}{3} & \multirow{2}{*}{1 (Th.~\ref{th:simple-chiDmin})} & \multirow{2}{*}{1 (Th.~\ref{th:paths})} & \multirow{2}{*}{2 (Th.~\ref{th:simple-chiDmin})} & \multirow{2}{*}{2 (Th.~\ref{th:paths})}  \\
   & $n\ge 2$ & & & & & & &  \\
\hline
\multirow{2}{*}{2.} & $P_{2n+1}$, & \multirow{2}{*}{2} & \multirow{2}{*}{2} & \multirow{2}{*}{2} & \multirow{2}{*}{1 (Th.~\ref{th:simple-chiDmin})} & \multirow{2}{*}{2 (Th.~\ref{th:paths})}  & \multirow{2}{*}{2 (Th.~\ref{th:simple-chiDmin})} & \multirow{2}{*}{2 (Th.~\ref{th:paths})}  \\
   & $n\ge 1$ & & & & & & &  \\
\hline
\hline
3. & $C_{4}$ & 2 & 3 & 4 & 1 (Th.~\ref{th:simple-chiDmin}) & 2 (Th.~\ref{th:cycles}) & 2 (Th.~\ref{th:simple-chiDmin}) & 3 (Th.~\ref{th:cycles}) \\
\hline
4. & $C_{5}$ & 3 & 3 & 3 & 1 (Th.~\ref{th:simple-chiDmin}) & 2 (Th.~\ref{th:cycles}) & 3 (Th.~\ref{th:simple-chiDmin}) & 3 (Th.~\ref{th:cycles}) \\
\hline
5. & $C_{6}$ & 2 & 2 & 4 & 1 (Th.~\ref{th:simple-chiDmin}) & 2 (Th.~\ref{th:cycles}) & 2 (Th.~\ref{th:simple-chiDmin}) & 3 (Th.~\ref{th:cycles}) \\
\hline
\multirow{2}{*}{6.} & $C_{2n}$, & \multirow{2}{*}{2} & \multirow{2}{*}{2} & \multirow{2}{*}{3} & \multirow{2}{*}{1 (Th.~\ref{th:simple-chiDmin})} & \multirow{2}{*}{2 (Th.~\ref{th:cycles})} & \multirow{2}{*}{2 (Th.~\ref{th:simple-chiDmin})} & \multirow{2}{*}{3 (Th.~\ref{th:cycles})} \\
   & $n\ge 4$ & & & & & & &  \\
\hline
\multirow{2}{*}{7.} & $C_{2n+1}$, & \multirow{2}{*}{3} & \multirow{2}{*}{2} & \multirow{2}{*}{3} & \multirow{2}{*}{1 (Th.~\ref{th:simple-chiDmin})} & \multirow{2}{*}{2 (Th.~\ref{th:cycles})} & \multirow{2}{*}{3 (Th.~\ref{th:simple-chiDmin})} & \multirow{2}{*}{3 (Th.~\ref{th:cycles})} \\
   & $n\ge 3$ & & & & & & &  \\
\hline
\hline
8. & $K_{3}$ & 3 & 3 & 3 & 1 (Th.~\ref{th:simple-chiDmin}) & 2 (Th.~\ref{th:complete-graphs}) & 3 (Th.~\ref{th:simple-chiDmin}) & 3 (Th.~\ref{th:complete-graphs}) \\
\hline
9. & $K_{4}$ & 3 & 3 & 5 & 1 (Th.~\ref{th:simple-chiDmin}) & 2 (Th.~\ref{th:complete-graphs}) & 3 (Th.~\ref{th:simple-chiDmin}) & 3 (Th.~\ref{th:complete-graphs})\\
\hline
10. & $K_{5}$ & 5 & 3 & 5 & 1 (Th.~\ref{th:simple-chiDmin}) & 2 (Th.~\ref{th:complete-graphs}) & 5 (Th.~\ref{th:simple-chiDmin}) & 5 (Th.~\ref{th:complete-graphs}) \\
\hline
\multirow{2}{*}{11.} & $K_{2n}$, & \multirow{2}{*}{$2n-1$} & \multirow{2}{*}{2} & \multirow{2}{*}{$2n-1$} & \multirow{2}{*}{1 (Th.~\ref{th:simple-chiDmin})}  & \multirow{2}{*}{2 (Th.~\ref{th:complete-graphs})} & $2n-1$ & $2n-1$ \\
   & $n\ge 3$ & & & &  &  & (Th.~\ref{th:simple-chiDmin}) & (Th.~\ref{th:complete-graphs}) \\
\hline
\multirow{2}{*}{12.} & $K_{2n+1}$, & \multirow{2}{*}{$2n+1$} & \multirow{2}{*}{2} & \multirow{2}{*}{$2n+1$} & \multirow{2}{*}{1 (Th.~\ref{th:simple-chiDmin})} & \multirow{2}{*}{2 (Th.~\ref{th:complete-graphs})} & $2n+1$ & $2n+1$ \\
   & $n\ge 3$ & & & & & & (Th.~\ref{th:simple-chiDmin}) & (Th.~\ref{th:complete-graphs}) \\
\hline
\hline
\multirow{2}{*}{13.} & $K_{1,n}$, & \multirow{2}{*}{$n$} & \multirow{2}{*}{$n$} & \multirow{2}{*}{$n$} & \multirow{2}{*}{$\left\lceil\frac{n}{2}\right\rceil$ (Th.~\ref{th:K_1n})} & \multirow{2}{*}{$n$ (Th.~\ref{th:K_1n})} & \multirow{2}{*}{$n$ (Th.~\ref{th:K_1n})} & \multirow{2}{*}{$n$ (Th.~\ref{th:K_1n})} \\
   & $n\ge 3$ & & & & & & &  \\
\hline
\multirow{2}{*}{14.} & \multirow{2}{*}{$K_{3,3}$}  & \multirow{2}{*}{$3$} & \multirow{2}{*}{$3$} & \multirow{2}{*}{$5$} & 1  & 2  & 3  & 4 \\
 & & & & & (Th.~\ref{th:K_nn}) & (Th.~\ref{th:K_nn}) & (Th.~\ref{th:K_nn}) & (Th.~\ref{th:K_nn})\\
\hline
\multirow{2}{*}{15.} & $K_{n,n}$,  & \multirow{2}{*}{$n$} & \multirow{2}{*}{$2$} & \multirow{2}{*}{$n+1$} & {1} & {2} & {$n$} & {$n+1$} \\
   & $4\le n\le 6$ & & & &  (Th.~\ref{th:K_nn})&  (Th.~\ref{th:K_nn})&  (Th.~\ref{th:K_nn})&  (Th.~\ref{th:K_nn}) \\
\hline
\multirow{2}{*}{16.} & $K_{n,n}$,  & \multirow{2}{*}{$n$} & \multirow{2}{*}{$2$} & \multirow{2}{*}{$n+1$} & {1} & {2} & {$n$} & {$n$} \\
   & $n\ge 7$ & & & &  (Th.~\ref{th:K_nn})&  (Th.~\ref{th:K_nn})&  (Th.~\ref{th:K_nn})&  (Th.~\ref{th:K_nn}) \\
\hline
\multirow{2}{*}{17.} & $K_{m,n}$,  & \multirow{2}{*}{$n$} & \multirow{2}{*}{\cite{FI08}, \cite{IJK08}} & \multirow{2}{*}{$n$} & Th.~\ref{th:complete-bipartite-graphs} and & $\DP(K_{m,n})$   & Th.~\ref{th:complete-bipartite-graphs} and & $n$  \\
   & $n > m\ge 2$ & & & & Cor.~\ref{cor:final-complete-bipartite} & (Th.~\ref{th:complete-bipartite-graphs}) & Cor.~\ref{cor:final-complete-bipartite} &  (Th.~\ref{th:complete-bipartite-graphs})\\
\hline
  \end{tabular}
\end{center}
}
\caption{Table of results for $\DminP(G)$, $\DmaxP(G)$,  $\chiDminP(G)$ and $\chiDmaxP(G)$.}
\label{table:global-edge}
\end{table}

Distinguishing numbers of digraphs have been studied in a few papers (see \cite{AC99,LNS10,LS12,L13,L19,MS17}),
while distinguishing numbers or indices of various classes of undirected graphs have attracted a lot of attention
(see for instance \cite{A05,AS17,AS19,ACD08,BC04,C09,EIKPT17,FI08,IJK08,IK06,KZ07}).
Up to our knowledge, distinguishing chromatic number, distinguishing index and distinguishing chromatic
index of digraphs have not been considered yet.

Our paper is organised as follows.
In Section~\ref{sec:basic}, we formally introduce the main definitions and give some basic results.
We then consider simple classes of graphs, namely paths, cycles and complete graphs in Section~\ref{sec:simple},
and complete bipartite graphs in Sections \ref{sec:complete-bipartite-easy}
and~\ref{sec:complete-bipartite-general}.
We finally propose some directions for future work in Section~\ref{sec:discussion}.

\medskip

Our results concerning distinguishing numbers and distinguishing indices are summarized
in Tables \ref{table:global-vertex} and~\ref{table:global-edge}, respectively,
where $P_n$, $C_n$ and~$K_n$ denote the path, the cycle and 
the complete graph of order~$n$, respectively,
and $K_{m,n}$ denotes the complete bipartite graph whose parts have size $m$ and~$n$.

\section{Preliminaries}
\label{sec:basic}

An \emph{oriented graph} is a digraph with no loops and no pairs of opposite arcs.
For an oriented graph $\OG$, we
denote by $V(\OG)$ and $A(\OG)$ its set of vertices and its set of arcs, respectively.
Let $\OG$ be an oriented graph and $u$ a vertex of $\OG$.
The \emph{out-degree} of $u$ in $\OG$, denoted $d^+_{\OG}(u)$, is the number of arcs in $A(\OG)$ of the form $uv$,
and the \emph{in-degree} of $u$ in $\OG$, denoted $d^-_{\OG}(u)$, is the number of arcs in~$A(\OG)$ of the form $vu$.
The \emph{degree} of $u$, denoted $d_{\OG}(u)$, is then defined by $d_{\OG}(u)=d_{\OG}^+(u)+d_{\OG}^-(u)$.
If $uv$ is an arc in $\OG$, $u$ is an \emph{in-neighbour} of $v$ and $v$ is an \emph{out-neighbour} of $u$.
We denote by $N_{\OG}^+(u)$ and $N_{\OG}^-(u)$ the set of out-neighbours and the set of in-neighbours of $u$ in ${\OG}$,
respectively. Hence, $d^+_{\OG}(u)=|N_{\OG}^+(u)|$ and $d^-_{\OG}(u)=|N_{\OG}^-(u)|$.
A \emph{source vertex} is a vertex with no in-neighbours, while
a \emph{sink vertex} is a vertex with no out-neighbours.
Let $v$ and $w$ be two neighbours of $u$. We say that $v$ and $w$ \emph{agree on $u$}
if either both $v$ and $w$ are in-neighbours of $u$, or both $v$ and $w$ are out-neighbours of $u$,
and that $v$ and $w$ \emph{disagree on $u$} otherwise.
%
For a subset $S$ of $V(\OG)$, we denote by $G[S]$ the sub-digraph of $\OG$
induced by $S$, which is obviously an oriented graph.

\medskip

For any finite set $\Omega$, $Id_\Omega$ denotes the identity permutation acting on $\Omega$.
Since the set $\Omega$ will always be clear from the context, we will simply write
$Id$ instead of $Id_\Omega$ in the following.

An \emph{automorphism} of an oriented graph $\OG$ is an arc-preserving permutation of its vertices, that is,
a one-to-one mapping $\phi:V(\OG)\rightarrow V(\OG)$ such that $\phi(u)\phi(v)$ is an arc in $\OG$ whenever $uv$
is an arc in~$\OG$.
The set of automorphisms of $\OG$ is denoted $\Aut(\OG)$.
The \emph{order} of an automorphism $\phi$ is the smallest integer $k>0$ for which $\phi^k=Id$.
An automorphism $\phi$ of an oriented graph $\OG$ is \emph{non-trivial} if $\phi\neq Id$. 
A vertex $u$ of $\OG$ is \emph{fixed} by $\phi$ if $\phi(u)=u$.

\medskip

We now introduce the distinguishing parameters we will consider.
Let $\lambda$ be a vertex-labelling of an oriented graph $\OG$.
Recall first that an automorphism $\phi$ of $\OG$ is \emph{$\lambda$-preserving} 
if $\lambda(\phi(u))=\lambda(u)$ for
every vertex $u$ of $\OG$, and that
a vertex-labelling  $\lambda$ of $\OG$ is \emph{distinguishing} if the only
$\lambda$-preserving automorphism of $\OG$ is the identity.
A distinguishing vertex-colouring is then a distinguishing proper vertex-labelling.
Similarly, an automorphism $\phi$ of an oriented graph~$\OG$ is \emph{$\lambda$-preserving}, 
for a given arc-labelling $\lambda$ of~$G$, if $\lambda(\phi(\OR{uv}))=\lambda(\OR{uv})$ for
every arc $\OR{uv}$ of $\OG$, and an arc-labelling  $\lambda$ of $\OG$ is \emph{distinguishing} 
if the only $\lambda$-preserving automorphism of $\OG$ is the identity.
A distinguishing arc-colouring is then a distinguishing proper arc-labelling.
We then define the four following distinguishing parameters of an oriented graph $\OG$.

\begin{enumerate}
\item The \emph{oriented distinguishing number} of $\OG$, denoted $\OD(\OG)$, is the smallest number
of labels needed for a distinguishing vertex-labelling of $\OG$.

\item The \emph{oriented distinguishing chromatic number} of $\OG$, denoted $\OchiD(\OG)$, is the smallest number
of labels needed for a distinguishing vertex-colouring of $\OG$.

\item The \emph{oriented distinguishing index} of $\OG$, denoted $\OD'(\OG)$, is the smallest number
of labels needed for a distinguishing arc-labelling of $\OG$.

\item The \emph{oriented distinguishing chromatic index} of $\OG$, denoted $\OchiDP(\OG)$, is the smallest number
of labels needed for a distinguishing arc-colouring of $\OG$.
\end{enumerate}

Using these parameters, we can define eight new distinguishing parameters of an undirected graph $G$.

\begin{enumerate}
\item The \emph{minimum oriented distinguishing number} of $G$, denoted $\Dmin(G)$, is the 
smallest oriented distinguishing number of its orientations.

\item The \emph{maximum oriented distinguishing number} of $G$, denoted $\Dmax(G)$, is the 
largest oriented distinguishing number of its orientations.

\item The \emph{minimum oriented distinguishing chromatic number} of $G$, denoted $\chiDmin(G)$, is the 
smallest oriented distinguishing chromatic number of its orientations.

\item The \emph{maximum oriented distinguishing chromatic number} of $G$, denoted $\chiDmax(G)$, is the 
largest oriented distinguishing chromatic number of its orientations.

\item The \emph{minimum oriented distinguishing index} of $G$, denoted $\DminP(G)$, is the 
smallest oriented distinguishing index of its orientations.

\item The \emph{maximum oriented distinguishing index} of $G$, denoted $\DmaxP(G)$, is the 
largest oriented distinguishing index of its orientations.

\item The \emph{minimum oriented distinguishing chromatic index} of $G$, denoted $\chiDminP(G)$, is the 
smallest oriented distinguishing chromatic index of its orientations.

\item The \emph{maximum oriented distinguishing chromatic index} of $G$, denoted $\chiDmaxP(G)$, is the 
largest oriented distinguishing chromatic index of its orientations.

\end{enumerate}

\medskip

Let $G$ be an undirected graph, $\OG$ be any orientation of $G$,
and $\lambda$ be any vertex or edge-labelling of $G$ (which can also be considered
as a vertex or arc-labelling of $\OG$, respectively).
Observe that every automorphism of $\OG$ is an automorphism of $G$.
From this observation and the definition of our distinguishing parameters,
we directly get the following result.

\begin{proposition}\label{prop:inequalities}
For every undirected graph $G$,
\begin{enumerate}
\item $\Dmin(G)\le \Dmax(G)\le D(G)$,
\item $\DminP(G)\le \DmaxP(G)\le D'(G)$,
\item $\chi(G)\le \chiDmin(G)\le \chiDmax(G)\le \chiD(G)$,
\item $\chi'(G)\le \chiDminP(G)\le \chiDmaxP(G)\le \chiDP(G)$,
\item $D(G)\le\chiD(G)$, $\Dmin(G)\le\chiDmin(G)$, and $\Dmax(G)\le\chiDmax(G)$,
\item $D'(G)\le\chiDP(G)$, $\DminP(G)\le\chiDminP(G)$, and $\DmaxP(G)\le\chiDmaxP(G)$.
\end{enumerate}
\end{proposition}

As observed in the previous section, for every undirected graph~$G$,
$D(G)=1$ if and only if $G$~is rigid, and 
$\chiD(G)=\chi(G)$ if $G$~is rigid. This property obviously also
holds for oriented graphs: for every oriented graph~$\OG$,
$\OD(\OG)=1$ if and only if $\OG$~is rigid, and 
$\OchiD(\OG)=\chi(\OG)$ if $\OG$~is rigid. 
We thus have the following result.

\begin{proposition}\label{prop:rigid-orientation}
If $G$ is an undirected graph that admits a rigid orientation,
then $\Dmin(G)=1$, $\DminP(G)=1$, $\chiDmin(G)=\chi(G)$ and $\chiDminP(G)=\chi'(G)$.
\end{proposition}

\section{Paths, cycles and complete graphs}
\label{sec:simple}

The distinguishing number of undirected paths, cycles
and complete graphs has been determined by Albertson and Collins in~\cite{AC96},
while their distinguishing chromatic number
was given by Collins and Trenk in~\cite{CT06}. 
The following theorem summarizes these results.

\begin{theorem}[\cite{AC96}, \cite{CT06}]\label{th:simple-vertex}\mbox{}
  \begin{enumerate}
  \item For every integer $n\ge 1$, $D(P_{2n})=\chiD(P_{2n})=2$.
  \item For every integer $n\ge 1$, $D(P_{2n+1})=2$ and $\chiD(P_{2n+1})=3$.
  \item $D(C_4)=3$ and $\chiD(C_4)=4$.
  \item $D(C_5)=3$ and $\chiD(C_5)=3$.
  \item $D(C_6)=2$ and $\chiD(C_6)=4$.
  \item For every integer $n\ge 7$, $D(C_{n})=2$ and $\chiD(C_{n})=3$.
  \item For every integer $n\ge 1$, $D(K_{n})=\chiD(K_{n})=n$.
  \end{enumerate}
\end{theorem}

On the other hand, the distinguishing index and distinguishing chromatic index of 
undirected paths, cycles
and complete graphs, have been determined by Kalinowski and Pil\'sniak in~\cite{KP15}, and
Alekhani and Soltani in~\cite{AS19}.
The following theorem summarizes these results.

\begin{theorem}[\cite{KP15}, \cite{AS19}]\label{th:simple-edge}\mbox{}
  \begin{enumerate}
  \item $\DP(P_{2})=\chiDP(P_{2})=1$.
  \item For every integer $n\ge 2$, $\DP(P_{2n})=2$ and $\chiDP(P_{2n})=3$.
  \item For every integer $n\ge 1$, $\DP(P_{2n+1})=\chiDP(P_{2n})=2$.
  \item $\DP(C_4)=3$ and $\chiDP(C_4)=4$.
  \item $\DP(C_5)=3$ and $\chiDP(C_5)=3$.
  \item $\DP(C_6)=2$ and $\chiDP(C_6)=4$.
  \item For every integer $n\ge 7$, $\DP(C_{n})=2$ and $\chiDP(C_{n})=3$.
  \item $\DP(K_3)=3$ and $\chiDP(K_3)=3$.
  \item $\DP(K_4)=\DP(K_5)=3$ and $\chiDP(K_4)=\chiDP(K_5)=5$.
  \item For every integer $n\ge 3$, $\DP(K_{2n})=2$ and $\chiDP(K_{2n})=2n-1$.
  \item For every integer $n\ge 3$, $\DP(K_{2n+1})=2$ and $\chiDP(K_{2n+1})=2n+1$.
  \end{enumerate}
\end{theorem}

It is not difficult to observe that every path, cycle or complete graph, admits
a rigid orientation. We thus get the following result, which proves columns
$\Dmin(G)$ and $\chiDmin(G)$ of Table~\ref{table:global-vertex} for lines 1 to~8,
and columns
$\DminP(G)$ and $\chiDminP(G)$ of Table~\ref{table:global-edge} for lines 1 to~12.
We say that an oriented path, or an oriented cycle, is \emph{directed} if all 
its arcs have the same direction.
 
\begin{theorem}\label{th:simple-chiDmin}
If $G$ is an undirected path, cycle, or complete graph, then $G$ admits a rigid
orientation.
Therefore, $\Dmin(G)=\DminP(G)=1$, $\chiDmin(G)=\chi(G)$ and $\chiDminP(G)=\chiP(G)$.
\end{theorem}

\begin{proof}
Since all directed paths and transitive tournaments are rigid, and
every orientation of the cycle $C_n$, $n\ge 3$, obtained from the directed cycle by
reversing exactly one arc is rigid, the first statement holds.
The second statement then directly follows from Proposition~\ref{prop:rigid-orientation}.
\end{proof}

We now consider the parameters
$\Dmax(G)$, $\chiDmax(G)$, $\DmaxP(G)$ and $\chiDmaxP(G)$ for
$G$ being a path, a cycle, or a complete graph.
For paths, we have the following result, which proves columns
$\Dmax(G)$ and $\chiDmax(G)$ of Table~\ref{table:global-vertex} 
and columns
$\DmaxP(G)$ and $\chiDmaxP(G)$ of Table~\ref{table:global-edge}, for lines 1 and~2.

\begin{theorem}\label{th:paths}
Let $P_n$ denote the path of order~$n$.
We then have
$\Dmax(P_n)=\DmaxP(P_n)=1$ and $\chiDmax(P_n)=\chiDmaxP(P_n)=2$
if $n$ is even, and
$\Dmax(P_n)=\DmaxP(P_n)=\chiDmaxP(P_n)=2$ and $\chiDmax(P_n)=3$
otherwise.
\end{theorem}

\begin{proof}
Let us denote by $P_n=v_1\dots v_n$, $n\ge 2$, the undirected path of order $n$. Note that
the only non-trivial automorphism of~$P_n$ is the permutation~$\pi$ that exchanges vertices $x_i$
and~$x_{n-i+1}$ for every~$i$, $1\le i\le\lfloor\frac{n}{2}\rfloor$.

If $n$ is even, every orientation of~$P_n$ is rigid, which gives 
$\Dmax(P_n)=\DmaxP(P_n)=1$ and $\chiDmax(P_n)=\chiDmaxP(P_n)=2$ by Proposition~\ref{prop:rigid-orientation}.

Suppose now that $n$ is odd.
In that case, every orientation such that the edges $v_iv_{i+1}$ and $v_{n-i}v_{n-i+1}$ have opposite
directions is not rigid, which implies that the value of the four considered distinguishing
parameters is strictly greater than~1.
Since $D(P_n)=\DP(P_n)=\chiDP(P_n)=2$, we get $\Dmax(P_n)=\DmaxP(P_n)=\chiDmaxP(P_n)=2$ 
by Proposition~\ref{prop:inequalities}.
Finally, since every 2-vertex-colouring of~$P_n$ is preserved by $\pi$, we have
$\chiDmax(P_n)>2$ and thus, since $\chiD(P_n)=3$, $\chiDmax(P_n)=3$
by Proposition~\ref{prop:inequalities}.
\end{proof}

For cycles, we have the following result, which proves columns
$\Dmax(G)$ and $\chiDmax(G)$ of Table~\ref{table:global-vertex}
and columns
$\DmaxP(G)$ and $\chiDmaxP(G)$ of Table~\ref{table:global-edge}, for lines 3 to~7
(the case of the 3-cycle is covered by Theorem~\ref{th:complete-graphs}).

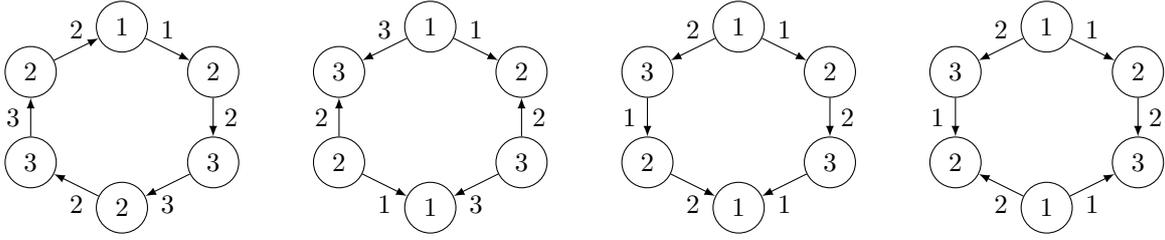
\begin{figure}
\begin{center}
\begin{tikzpicture}[scale=0.6]   
\node[draw,circle] (A) at (0,4) {{\footnotesize 1}};
\node[draw,circle] (B) at (2,3) {{\footnotesize 2}};
\node[draw,circle] (C) at (2,1) {{\footnotesize 3}};
\node[draw,circle] (D) at (0,0) {{\footnotesize 2}};
\node[draw,circle] (E) at (-2,1) {{\footnotesize 3}};
\node[draw,circle] (F) at (-2,3) {{\footnotesize 2}};
\draw[->,>=latex] (A) to (B);
\draw[->,>=latex] (B) to (C);
\draw[->,>=latex] (C) to (D);
\draw[->,>=latex] (D) to (E);
\draw[->,>=latex] (E) to (F);
\draw[->,>=latex] (F) to (A);
\node[above] at (1,3.5) {{\footnotesize 1}};
\node[right] at (2,2) {{\footnotesize 2}};
\node[below] at (1,0.5) {{\footnotesize 3}};
\node[below] at (-1,0.5) {{\footnotesize 2}};
\node[left] at (-2,2) {{\footnotesize 3}};
\node[above] at (-1,3.5) {{\footnotesize 2}};
%
\end{tikzpicture}
\hskip 0.6cm
\begin{tikzpicture}[scale=0.6]   
\node[draw,circle] (A) at (0,4) {{\footnotesize 1}};
\node[draw,circle] (B) at (2,3) {{\footnotesize 2}};
\node[draw,circle] (C) at (2,1) {{\footnotesize 3}};
\node[draw,circle] (D) at (0,0) {{\footnotesize 1}};
\node[draw,circle] (E) at (-2,1) {{\footnotesize 2}};
\node[draw,circle] (F) at (-2,3) {{\footnotesize 3}};
\draw[->,>=latex] (A) to (B);
\draw[->,>=latex] (C) to (B);
\draw[->,>=latex] (C) to (D);
\draw[->,>=latex] (E) to (D);
\draw[->,>=latex] (E) to (F);
\draw[->,>=latex] (A) to (F);
\node[above] at (1,3.5) {{\footnotesize 1}};
\node[right] at (2,2) {{\footnotesize 2}};
\node[below] at (1,0.5) {{\footnotesize 3}};
\node[below] at (-1,0.5) {{\footnotesize 1}};
\node[left] at (-2,2) {{\footnotesize 2}};
\node[above] at (-1,3.5) {{\footnotesize 3}};
%
\end{tikzpicture}
\hskip 0.6cm
\begin{tikzpicture}[scale=0.6]   
\node[draw,circle] (A) at (0,4) {{\footnotesize 1}};
\node[draw,circle] (B) at (2,3) {{\footnotesize 2}};
\node[draw,circle] (C) at (2,1) {{\footnotesize 3}};
\node[draw,circle] (D) at (0,0) {{\footnotesize 1}};
\node[draw,circle] (E) at (-2,1) {{\footnotesize 2}};
\node[draw,circle] (F) at (-2,3) {{\footnotesize 3}};
\draw[->,>=latex] (A) to (B);
\draw[->,>=latex] (B) to (C);
\draw[->,>=latex] (C) to (D);
\draw[->,>=latex] (E) to (D);
\draw[->,>=latex] (F) to (E);
\draw[->,>=latex] (A) to (F);
\node[above] at (1,3.5) {{\footnotesize 1}};
\node[right] at (2,2) {{\footnotesize 2}};
\node[below] at (1,0.5) {{\footnotesize 1}};
\node[below] at (-1,0.5) {{\footnotesize 2}};
\node[left] at (-2,2) {{\footnotesize 1}};
\node[above] at (-1,3.5) {{\footnotesize 2}};
%
\end{tikzpicture}
\hskip 0.6cm
\begin{tikzpicture}[scale=0.6]   
\node[draw,circle] (A) at (0,4) {{\footnotesize 1}};
\node[draw,circle] (B) at (2,3) {{\footnotesize 2}};
\node[draw,circle] (C) at (2,1) {{\footnotesize 3}};
\node[draw,circle] (D) at (0,0) {{\footnotesize 1}};
\node[draw,circle] (E) at (-2,1) {{\footnotesize 2}};
\node[draw,circle] (F) at (-2,3) {{\footnotesize 3}};
\draw[->,>=latex] (A) to (B);
\draw[->,>=latex] (B) to (C);
\draw[->,>=latex] (D) to (C);
\draw[->,>=latex] (D) to (E);
\draw[->,>=latex] (F) to (E);
\draw[->,>=latex] (A) to (F);
\node[above] at (1,3.5) {{\footnotesize 1}};
\node[right] at (2,2) {{\footnotesize 2}};
\node[below] at (1,0.5) {{\footnotesize 1}};
\node[below] at (-1,0.5) {{\footnotesize 2}};
\node[left] at (-2,2) {{\footnotesize 1}};
\node[above] at (-1,3.5) {{\footnotesize 2}};
%
\end{tikzpicture}

\caption{Distinguishing vertex and edge-colourings of non-rigid orientations of $C_6$.\label{fig:C6}}
\end{center}
\end{figure}

\begin{theorem}\label{th:cycles}\mbox{}
\begin{enumerate}
\item For every integer $n\ge 4$, $\Dmax(C_n)=\DmaxP(C_n)=2$.
\item $\chiDmax(C_4)=4$ and, for every integer $n\ge 5$, $\chiDmax(C_n)=3$.
\item For every integer $n\ge 4$, $\chiDmaxP(C_n)=3$.
\end{enumerate}
\end{theorem}

\begin{proof}
Let us denote by $C_n=v_1\dots v_nv_1$, $n\ge 4$, the undirected cycle of order $n$.
Since the directed cycle $\OR{C_n}$ is not rigid, the value of the four considered distinguishing
parameters is strictly greater than~1.

Let us first consider distinguishing labellings and let $\OR{C}$ be any orientation of $C_n$.
The 2-vertex-labelling $\lambda$ of $\OR{C}$ defined by
$\lambda(v_1)=\lambda(v_2)=1$ and $\lambda(v_i)=2$ for every $i$, $3\le i\le n$,
is clearly distinguishing since $v_1$ and $v_2$
must be fixed by every $\lambda$-preserving automorphism of $\OR{C}$.
We thus get $\Dmax(C_n)=2$ for every $n\ge 4$.
Similarly, the 2-edge-labelling $\lambda'$ of $\OR{C}$ that assigns colour~1 to exactly one arc of $\OR{C}$
is distinguishing since the end-vertices of the arc coloured with~1 
must be fixed by every $\lambda$-preserving automorphism of $\OR{C}$. 
Therefore, $\DmaxP(C_n)=2$ for every $n\ge 4$.

Let us now consider distinguishing colourings.
Suppose first that $n$ is odd. 
In that case, since $\chi(C_n)=\chiP(C_n)=3$, we get $\chiDmax(C_n)=\chiD(C_n)=3$
and $\chiDmaxP(C_n)=\chiDP(C_n)=3$
by Proposition~\ref{prop:inequalities} and Theorem \ref{th:simple-vertex} or~\ref{th:simple-edge}.

Suppose now that $n$ is even.
Note that every 2-vertex or 2-arc-colouring of the directed cycle $\OR{C_n}$ is $\rho^2$-preserving, 
where $\rho$ denotes the automorphism of $\OR{C_n}$ defined by $\rho(v_i)=v_{i+1}$ for every~$i$,
$1\le i\le n-1$, and $\rho(v_n)=v_1$. This implies $\chiDmax(C_n)>2$ and $\chiDmaxP(C_n)>2$.
We thus get 
$\chiDmax(C_n)=\chiD(C_n)=3$
and $\chiDmaxP(C_n)=\chiDP(C_n)=3$
for every even $n\ge 8$, thanks to
Proposition~\ref{prop:inequalities} and Theorem \ref{th:simple-vertex} or~\ref{th:simple-edge}.

Consider now the cycle $C_4=v_1v_2v_3v_4v_1$. Note that, up to symmetries, $C_4$ has only three
non-rigid orientations, namely the directed cycle $\OR{C_4}$, the orientation $\OR{C'_4}$
with arcs $\OR{v_1v_2}$, $\OR{v_3v_2}$, $\OR{v_3v_4}$ and $\OR{v_1v_4}$,
and the orientation $\OR{C''_4}$
with arcs $\OR{v_1v_2}$, $\OR{v_3v_2}$, $\OR{v_4v_1}$ and $\OR{v_4v_3}$.
Moreover, observe that $\Aut(\OR{C'_4})$ has three non-trivial automorphisms,
one exchanging $v_1$ and $v_3$, one exchanging $v_2$ and $v_4$, and the third one
exchanging both these pairs of vertices. This implies that every distinguishing 
vertex-colouring of $\OR{C'_4}$ must use four colours, and thus
$\chiDmax(C_4)=4$.
To see that $\chiDmaxP(C_4)=3$, it suffices to observe that, by colouring the arcs of 
$\OR{C_4}$, $\OR{C'_4}$ or $\OR{C''_4}$ cyclically with colours $1213$, we obtain a distinguishing 3-arc-colouring.

Consider finally the cycle~$C_6$. 
Up to symmetries, $C_6$ has four distinct non-rigid orientations,
depicted in Figure~\ref{fig:C6}.
It is not difficult to check that the vertex and arc-colourings described
in the same figure are all optimal, which gives 
$\chiDmax(C_6)=\chiDmaxP(C_6)=3$.

This concludes the proof.
\end{proof}

The distinguishing number of tournaments has been studied by Albertson and
Collins in~\cite{AC99}, where it is proved in particular that 
$D(T_n)\le 1 + \left\lceil\frac{\lceil\log n\rceil}{2}\right\rceil$ for every tournament $T_n$ of order~$n$.
Moreover, they conjectured that $D(T)\le 2$ for every tournament $T$.
As observed by Godsil in 2002 (this fact is mentioned by Lozano in~\cite{L19}), this conjecture
follows from Gluck's Theorem~\cite{G83}.

For every integer $n\ge 3$, the values of
$\Dmin(K_n)$, $\chiDmin(K_n)$, $\DminP(K_n)$ and $\chiDminP(K_n)$
are already given by Theorem~\ref{th:simple-chiDmin}.
The values of the other parameters are given by the following result,
which proves columns $\Dmax(G)$ and $\chiDmax(G)$ of Table~\ref{table:global-vertex}, line~8, 
and columns $\DmaxP(G)$ and $\chiDmaxP(G)$ of Table~\ref{table:global-edge}, lines 8 to~12.

\begin{theorem}\label{th:complete-graphs}
For every integer $n\ge 3$,
$\Dmax(K_n)=\DmaxP(K_n)=2$, $\chiDmax(K_n)=n$, and $\chiDmaxP(K_n)=\chiP(K_n)$.
\end{theorem}

\begin{proof}
Since every complete graph $K_n$, $n\ge 3$, admits a non rigid orientation
(consider for instance the orientation $\OR{K_n}$ of $K_n$
obtained from the transitive orientation of $K_n$ with directed path
$x_1\dots x_n$, by reversing the arc $x_{n-2}x_n$, so that $x_{n-2}$, $x_{n-1}$ and $x_n$
form a directed cycle; the permutation $(x_{n-2},x_{n-1},x_n)$ is then clearly a 
non trivial automorphism of $\OR{K_n}$), 
we get $\Dmax(K_n)\ge 2$ and $\DmaxP(K_n)\ge 2$.

Since, as mentioned above, $D(T)\le 2$ for every tournament $T$,
we get, by Proposition~\ref{prop:inequalities}, 
$\Dmax(K_n)\le 2$ and thus $\Dmax(K_n)=2$.
By Proposition~\ref{prop:inequalities}, we also have
$\DmaxP(K_n)\le \DP(K_n)$, and thus $\DmaxP(K_n)=2$ if $n\ge 6$.
If $n=3$, then the only non rigid orientation of $K_3$ is the
directed cycle. By assigning colour~1 to one arc and colour~2 to the
two other arcs, we get a distinguishing arc-labelling, so that $\Dmax(K_3)=2$.
If $n\in\{4,5\}$, observe that each non rigid orientation of $K_n$ contains
a transitive triangle.
By assigning colour~1 to the three arcs of this triangle and colour~2 to all the
other arcs, we get a distinguishing arc-labelling, so that $\DmaxP(K_4)=\DmaxP(K_5)=2$.

Finally, by Proposition~\ref{prop:inequalities}, 
$\chi(K_n)\le \chiDmax(K_n)\le \chiD(K_n)$ and 
$\chiP(K_n)\le \chiDmaxP(K_n)\le \chiDP(K_n)$,
which gives, thanks to Theorems \ref{th:simple-vertex} and~\ref{th:simple-edge},
$\chiDmax(K_n)=n$, and $\chiDmaxP(K_n)=\chiP(K_n)$ if $n\neq 4$, respectively.
Now, there are only two orientations of $K_4$ that admit a non trivial automorphism,
one with a source vertex and the other one with a sink vertex,
and the three other vertices in both of them inducing a directed 3-cycle.
In each case, every 3-edge-colouring is clearly distinguishing,
which gives $\chiDmaxP(K_4)=3=\chiP(K_4)$.
\end{proof}

 
\section{Complete bipartite graphs: easy cases}
\label{sec:complete-bipartite-easy}

In this section, we consider ``easy cases'' of complete bipartite graphs,
namely $K_{1,n}$ and $K_{n,n}$ for every $n\ge 3$ (the cases $n\in\{1,2\}$ correspond to
$P_2$, $P_3$ or $C_4$, already considered in the previous section).

Concerning stars $K_{1,n}$, it is not difficult to get the following result,
which proves line 9 of Table~\ref{table:global-vertex} and 
line 13 of Table~\ref{table:global-edge}.

\begin{theorem}\label{th:K_1n}
For every integer $n\ge 3$,
$\Dmin(K_{1,n})=\DminP(K_{1,n})=\left\lceil\frac{n}{2}\right\rceil$, 
$\Dmax(K_{1,n})=\DmaxP(K_{1,n})=n$, 
$\chiDmin(K_{1,n})=1+\left\lceil\frac{n}{2}\right\rceil$,
$\chiDmax(K_{1,n})=n+1$,
and $\chiDminP(K_{1,n})=\chiDmaxP(K_{1,n})=n$.
\end{theorem}

\begin{proof}
Let $x$ be the central vertex of $K_{1,n}$ and $\{y_1,\dots,y_n\}$ the set of neighbours of $x$.
Let $\OR{K}$ denote any orientation of $K_{1,n}$.
Observe that every automorphism of $\OR{K}$ can only permute in-neighbours of $x$,
and out-neighbours of $x$, $x$ being fixed.
Hence, by considering the orientation of $K_{1,n}$ for which $d^+(x)=n$,
we get $\Dmax(K_{1,n})=\DmaxP(K_{1,n})=n$.

On one other hand, the minimum value of $\Dmin(K_{1,n})$ and $\DminP(K_{1,n})$ is
attained when the number of in-neighbours and the number of out-neighbours of $x$ differ
by at most one, which gives $\Dmin(K_{1,n})=\DminP(K_{1,n})=\left\lceil\frac{n}{2}\right\rceil$.
Similarly, since all in-neighbours (resp. all out-neighbours) must have distinct colours
in every distinguishing vertex-colouring of $K_{1,n}$, and these colours have to be
distinct from the colour of $x$, we get 
$\chiDmin(K_{1,n})=1+\left\lceil\frac{n}{2}\right\rceil$ (when 
the number of in-neighbours and out-neighbours of $x$ differ
by at most one),
and $\chiDmax(K_{1,n})=n+1$ (when all arcs are out-going arcs from $x$).

Finally, since all arc-colourings of every orientation of $K_{1,n}$
are distinguishing,
we get $\chiDminP(K_{1,n})=\chiDmaxP(K_{1,n})=\chiP(K_{1,n})=n$.
\end{proof}

The distinguishing number, distinguishing chromatic number, distinguishing index and dis-
tinguishing chromatic index
of balanced complete bipartite graphs $K_{n,n}$ have been studied
by Collins and Trenk~\cite{CT06},
Kalinowski and Pil\'sniak~\cite{KP15},
and Alikhani and Soltani~\cite{AS19}.
The following theorem summarizes these results.

\begin{theorem}[\cite{CT06}, \cite{KP15}, \cite{AS19}]\label{th:complete-bipartite-DPnn}\mbox{}
\begin{enumerate}
\item $D(K_{3,3})=4$, $\chiD(K_{3,3})=6$, $\DP(K_{3,3})=3$ and $\chiDP(K_{3,3})=5$.  
\item For every $n\ge 4$, $D(K_{n,n})=n+1$, $\chiD(K_{n,n})=2n$, $\DP(K_{n,n})=2$ and $\chiDP(K_{n,n})=n+1$.
\end{enumerate}
\end{theorem}

Let us denote by $X=\{x_1,\dots,x_n\}$ and $Y=\{y_1,\dots,y_n\}$ the bipartition
of $K_{n,n}$.
It is a well-known fact the every proper edge-colouring of $K_{n,n}$ corresponds
to a Latin square of order $n$, the colour of the edge $x_iy_j$, $1\le i,j\le n$,
corresponding to the symbol in row $i$ and column $j$.
A Latin square $L$ is {\em asymmetric}~\cite{S13} if, for every three permutations 
$\alpha$, $\beta$ and $\gamma$ on the set $\{1,\dots,n\}$,
the Latin square $L'$ defined by $L'[\alpha(i),\beta(j)]=\gamma(L[i,j])$ for every $i,j$,
$1\le i,j\le n$ is equal to $L$ if and only if $\alpha$, $\beta$ and $\gamma$ are
all the identity.
In other words, if an edge-colouring $\lambda$ of $K_{n,n}$ corresponds to an asymmetric
Latin square, then the only $\lambda$-preserving automorphism of $K_{n,n}$ that preserves  
the bipartition is the identity.
%
Therefore, if $\OR{K}$ is the orientation of $K_{n,n}$ defined by
$N_{\OR{K}}^+(x_i) = Y$ for every $i$, $1\le i\le n$, 
and $\lambda$ is any $n$-arc-colouring of $\OR{K}$ corresponding to an asymmetric
Latin square, then none of the non-trivial automorphisms of $\OR{K}$ is $\lambda$-preserving.

The following result, due to K.T.~Phelps~\cite{P80}, concerns asymmetric Latin squares and  
will be useful for our next result.

\begin{theorem}[\cite{P80}]\label{th:Phelps}
For every integer $n\ge 7$, there exists an asymmetric Latin square of order~$n$.
Moreover, the smallest asymmetric Latin squares have order~7.
\end{theorem}

For balanced complete bipartite graphs, we have the following result,
which proves line 10 of Table~\ref{table:global-vertex} and 
lines 14, 15 and 16 of Table~\ref{table:global-edge}.

\begin{theorem}\label{th:K_nn}
For every integer $n\ge 3$,
$\Dmin(K_{n,n})=\DminP(K_{n,n})=1$, $\Dmax(K_{n,n})=n$, $\DmaxP(K_{n,n})=2$, 
$\chiDmin(K_{n,n})=2$, $\chiDminP(K_{n,n})=n$ and $\chiDmax(K_{n,n})=2n$.
Moreover, $\chiDmaxP(K_{n,n})=n+1$ if $3\le n\le 6$, and $\chiDmaxP(K_{n,n})=n$ if $n\ge 7$.
\end{theorem}

\begin{proof}
Let $X=\{x_1,\dots,x_n\}$ and $Y=\{y_1,\dots,y_n\}$ denote the bipartition of $K_{n,n}$, $n\ge 3$.

Observe first that the orientation of $K_{n,n}$ given by $N^+(x_i)=\{y_j:\ j\le i\}$ 
for every $i$, $1\le i\le n$, is rigid,
since all vertices in $X$ have distinct in-degrees.
Therefore, thanks to Proposition~\ref{prop:rigid-orientation},
we get $\Dmin(K_{n,n})=\DminP(K_{n,n})=1$,
$\chiDmin(K_{n,n})=\chi(K_{n,n})=2$ and $\chiDminP(K_{n,n})=\chiP(K_{n,n})=n$
for every $n\ge 3$.

Consider now the orientation $\OR{K}$ of $K_{n,n}$ given by $N^+(x_i)=Y$ 
for every $i$, $1\le i\le n$.
Clearly, for every two permutations $\pi_X$ of $X$ and $\pi_Y$ of $Y$,
the product $\pi_X\pi_Y$ is an automorphism of $\OR{K}$. Therefore,
all vertices of $\OR{K}$ must be assigned distinct colours by every distinguishing
vertex colouring, which gives $\chiDmax(K_{n,n})=2n$.

The $n$-vertex-labelling $\varphi$ of $K_{n,n}$ defined by $\varphi(x_i)=\varphi(y_i)=i$
for every $i$, $1\le i\le n$, is distinguishing for every orientation of $K_{n,n}$
since every label is used on two vertices that are connected by an arc.
Hence, $\Dmax(K_{n,n})\le n$. For the orientation $\OR{K}$ of $K_{n,n}$ defined
in the previous paragraph, every distinguishing vertex-labelling must assign distinct
labels to every vertex in each part, and thus $\Dmax(K_{n,n})=n$.

Consider now the 2-edge-labelling $\varphi'$ of $K_{n,n}$ defined
by $\varphi'(x_iy_j)=1$ if and only if $j<i$, for every $i,j$, $1\le i,j\le n$.
Note that every two vertices in each part have a distinct number of incident edges with label~1.
Therefore, reasoning similarly as in the previous paragraph, we get $\DmaxP(K_{n,n})=2$.

If $n\ge 7$, it follows from Theorem~\ref{th:Phelps} that 
the orientation $\OR{K}$ of $K_{n,n}$ admits a distinguishing $n$-arc-colouring,
and thus $\chiDmaxP(K_{n,n})= \chiP(K_{n,n}) = n$.
Finally, if $3\le n\le 6$, 
consider the $(n+1)$-edge-colouring $\lambda$ of $K_{n,n}$ defined by
$\lambda(x_iy_j)=i+j\pmod{n+1}$ for every $i$, $j$, $1\le i,j\le n$.
Since, in each part, the set of colours used on the incident edges of every two vertices 
are distinct (the colour $i$ does not appear on the edges incident with $x_i$,
and the colour $j$ does not appear on the edges incident with $y_j$), 
we get that $\lambda$ is a distinguishing arc-colouring of every orientation
of $K_{n,n}$, and thus $\chiDmaxP(K_{n,n})\le n+1$.
On the other hand, it follows from Theorem~\ref{th:Phelps} 
that for every $n$-edge-colouring $\lambda$ of $K_{n,n}$, there
exists a non trivial $\lambda$-preserving automorphism of $K_{n,n}$ 
such that every vertex is mapped to a vertex in the same part.
Considering the orientation $\OR{K}$ of $K_{n,n}$ defined above, 
this gives $\chiDmaxP(K_{n,n})>n$, and thus $\chiDmaxP(K_{n,n})= n+1$.

This completes the proof.
\end{proof}

 
\section{Complete bipartite graphs: other cases}
\label{sec:complete-bipartite-general}

We consider in this section the remaining cases of complete bipartite graphs, that is,
unbalanced complete bipartite graphs $K_{m,n}$ with $2\le m<n$.

The distinguishing number and the distinguishing chromatic number of 
unbalanced complete bipartite graphs have been determined 
by Collins and Trenk~\cite{CT06}, while 
their distinguishing chromatic index has been given by
Alekhani and Soltani~\cite{AS19}.
The following theorem summarizes these results.

\begin{theorem}[\cite{CT06}, \cite{AS19}]\label{th:complete-bipartite}
For every two integers $m$ and $n$, $2\le m<n$,
  $$D(K_{m,n})=n,\ \chiD(K_{m,n})=m+n \mbox{ and } \chiDP(K_{m,n})=n.$$
\end{theorem}

As recalled in the previous section, the distinguishing index of 
balanced complete bipartite graphs $K_{n,n}$ has been given
by Kalinowski and Pil\'sniak~\cite{KP15} (see Theorem~\ref{th:complete-bipartite-DPnn}).
The distinguishing index of complete bipartite graphs $K_{m,n}$ with $n>m$ is less
easy to determine. This has been done
by Fisher and Isaak~\cite{FI08}, and independently by Imrich, Jerebic and Klav\v zar~\cite{IJK08}.

\begin{theorem}[\cite{FI08}, \cite{IJK08}]\label{th:complete-bipartite-DP}\mbox{}
Let $m$, $n$ and $r$ be integers such that $r\ge 2$ and $(r-1)^m < n \le r^m$.
We then have
$$\DP(K_{m,n}) = \left\{
   \begin{array}{ll}
   r, & \mbox{if $n\le r^m - \lceil \log_r m\rceil - 1$,} \\
   r+1, & \mbox{if $n\ge r^m - \lceil \log_r m\rceil + 1$.}
   \end{array}\right.
$$
Moreover, if $n = r^m - \lceil \log_r m\rceil$, then $\DP(K_{m,n})$ is either $r$
or $r+1$ and can be computed recursively in time $O(\log^*(n))$.
\end{theorem}

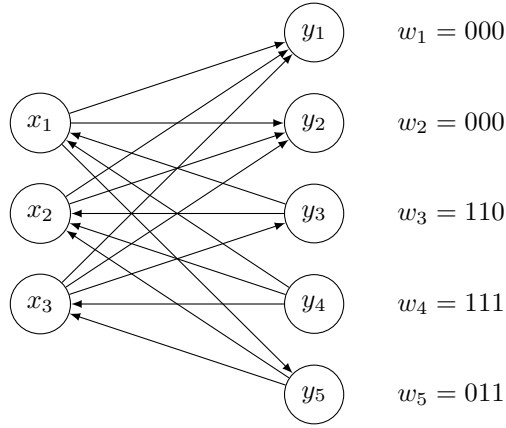
\begin{figure}
\begin{center}
\begin{tikzpicture}[scale=0.6]   
\node[draw,circle] (X1) at (0,2) {{\footnotesize $x_1$}};
\node[draw,circle] (X2) at (0,0) {{\footnotesize $x_2$}};
\node[draw,circle] (X3) at (0,-2) {{\footnotesize $x_3$}};
\node[draw,circle] (Y1) at (6,4) {{\footnotesize $y_1$}};
\node[draw,circle] (Y2) at (6,2) {{\footnotesize $y_2$}};
\node[draw,circle] (Y3) at (6,0) {{\footnotesize $y_3$}};
\node[draw,circle] (Y4) at (6,-2) {{\footnotesize $y_4$}};
\node[draw,circle] (Y5) at (6,-4) {{\footnotesize $y_5$}};
\node at (9,4) {{\footnotesize $w_1=000$}};
\node at (9,2) {{\footnotesize $w_2=000$}};
\node at (9,0) {{\footnotesize $w_3=110$}};
\node at (9,-2) {{\footnotesize $w_4=111$}};
\node at (9,-4) {{\footnotesize $w_5=011$}};
\draw[->,>=latex] (X1) to (Y1);
\draw[->,>=latex] (X2) to (Y1);
\draw[->,>=latex] (X3) to (Y1);
\draw[->,>=latex] (X1) to (Y2);
\draw[->,>=latex] (X2) to (Y2);
\draw[->,>=latex] (X3) to (Y2);
\draw[->,>=latex] (Y3) to (X1);
\draw[->,>=latex] (Y3) to (X2);
\draw[->,>=latex] (X3) to (Y3);
\draw[->,>=latex] (Y4) to (X1);
\draw[->,>=latex] (Y4) to (X2);
\draw[->,>=latex] (Y4) to (X3);
\draw[->,>=latex] (X1) to (Y5);
\draw[->,>=latex] (Y5) to (X2);
\draw[->,>=latex] (Y5) to (X3);
\end{tikzpicture}
\caption{An orientation of the complete bipartite graph $K_{3,5}$.\label{fig:orientation-K35}}
\end{center}
\end{figure}

Before considering the eight distinguishing parameters for complete
bipartite graphs we are interested in,
we introduce some notation and definitions.
For the complete bipartite graph $K_{m,n}$, $m<n$, we denote by $(X,Y)$ the partition
of its vertex set and let $X=\{x_1,\dots,x_m\}$ and $Y=\{y_1,\dots,y_n\}$.
Moreover, for every orientation $\OR{K}$ of $K_{m,n}$
and every integer $d$, $0\le d\le m$, we will denote by $Y_d$
the set of vertices in $Y$ with out-degree $d$, that is,
$Y_d=\{y\in Y\ |\ d^+_{\OR{K}}(y)=d\}$.
For example, in the orientation of $K_{3,5}$ depicted in Figure~\ref{fig:orientation-K35},
we have $Y_0=\{y_1,y_2\}$, $Y_1=\varnothing$, $Y_2=\{y_3,y_5\}$ and $Y_3=\{y_4\}$.


For general complete bipartite graphs, we have the following result.

\begin{theorem}\label{th:complete-bipartite-graphs}
For every two integers $m$ and $n$, $2\le m < n$, the following holds.
\begin{enumerate}
\item $\Dmax(K_{m,n})=\chiDmaxP(K_{m,n})=n$, $\DmaxP(K_{m,n}) = \DP(K_{m,n})$ 
and $\chiDmax(K_{m,n}) = m+n$.
\item If $K_{m,n}$ admits a rigid orientation, then
$\Dmin(K_{m,n}) = \DminP(K_{m,n}) = 1$, $\chiDmin(K_{m,n}) = 2$ and $\chiDminP(K_{m,n}) = n$.
\item If $K_{m,n}$ does not admit any rigid orientation, then
$\Dmin(K_{m,n}) \le \big\lceil\frac{n}{m-1}\big\rceil$,
$\chiDmin(K_{m,n}) \le 1 + \big\lceil\frac{n}{m-1}\big\rceil$,
$\DminP(K_{m,n}) \le \DP(K_{m,\big\lceil\frac{n}{m-1}\big\rceil})$ and
$\chiDminP(K_{m,n}) = n$.
\end{enumerate}
\end{theorem}

\begin{proof}
Consider the orientation $\OR{K}$ of $K_{m,n}$ given by $N^+(x_i)=Y$ for
every $i$, $1\le i\le m$.
In that case, since $m\neq n$, we have $Aut(\OR{K}) = Aut(K_{m,n})$,
which implies
$\Dmax(K_{m,n})=D(K_{m,n})$, $\chiDmaxP(K_{m,n})=\chiDP(K_{m,n})$, 
$\DmaxP(K_{m,n}) = \DP(K_{m,n})$ and $\chiDmax(K_{m,n}) = \chiD(K_{m,n})$.
Equalities in the first item then follows from Theorem~\ref{th:complete-bipartite}.
On one other hand,
equalities in the second item directly follow from Proposition~\ref{prop:rigid-orientation}.

\medskip

Suppose now that $K_{m,n}$ does not admit any rigid orientation.
We first claim that we necessarily have $m < \big\lceil\frac{n}{m-1}\big\rceil$.
Indeed, this follows from Lemmas \ref{lem:K2n}, \ref{lem:K3n-K4n} and~\ref{lem:m-over-2-to-2m-m},
proved later: 
if $m=2$, $3$ or $4$, then $n$ is at least $4$, $7$ or $14$,
respectively (see Lemmas \ref{lem:K2n} and~\ref{lem:K3n-K4n}), and the claim holds,
while $n>2^m-\big\lceil\frac{m}{2}\big\rceil$ if $m\ge 5$
(see Lemma~\ref{lem:m-over-2-to-2m-m}).
It is then easy to check that, for every $m\ge 5$,
$$\left\lceil\frac{n}{m-1}\right\rceil > \left\lceil\frac{2^m-\big\lceil\frac{m}{2}\big\rceil}{m-1}\right\rceil > m.$$

For every non-trivial automorphism $\phi$ of $K_{m,n}$, and every two vertices
$u$ and $v$ of $K_{m,n}$ with $\phi(u)=v$, we have $d^+(u)=d^+(v)$, and either
$u,v\in X$ or $u,v\in Y$.
Consider now any orientation $\OR{K}'$ of $K_{m,n}$ 
satisfying the two following properties:
\begin{itemize}
\item[P1.] $d_{\OR{K}'}^+(y_j)=1 + (j \mod m-1)$ for every $j$, $1\le j\le n$, and
\item[P2.] $N_{\OR{K}'}^+(y_{i(m-1)}) = \{x_i\}$ for every $i$, $1\le i\le m$.
\end{itemize}

It is not difficult to observe that such orientations necessarily exist.
Indeed, if $j\le m(m-1)$ and $j\equiv 0\pmod{m-1}$, $N_{\OR{K}'}^+(y_{j})$ is given by Property~P2.
Otherwise, pick arbitrarily a set $S_j$ of $1 + (j \mod m-1)$ vertices of $X$
(note that $1 + (j \mod m-1)<m$) and
let $N_{\OR{K}'}^+(y_{j})=S_j$.

Now, thanks to Property~P1, we have
$1\le d_{\OR{K}'}^+(y_j)\le m-1$ for every~$j$, $1\le j\le n$,
and the cardinality of each set $Y_d$, $1\le d\le m-1$, is at most $\big\lceil\frac{n}{m-1}\big\rceil$.
Moreover, Property~P2 ensures that 
the sets of in-neighbours of every two vertices in $X$ are distinct,
which implies
that every automorphism of $\OR{K}'$ that permutes vertices in $X$ must also permute vertices in $Y$.

Consider now the
$\big\lceil\frac{n}{m-1}\big\rceil$-vertex-labelling $\lambda$ of $\OR{K}'$ defined by
 $\lambda(x_i)=1$ for every $i$, $1\le i\le m$,
and $\lambda(y_j)=1+\big\lfloor\frac{j-1}{m-1}\big\rfloor$ for every $j$, $1\le j\le n$.
This labelling assigns 
distinct labels to the vertices of each set $Y_d$, $1\le d\le m-1$,
and is thus distinguishing, which gives $\Dmin(K_{m,n}) \le \big\lceil\frac{n}{m-1}\big\rceil$.
Moreover, by using one additional label for the vertices of~$X$ instead of label~1, the  
labelling $\lambda$ becomes a distinguishing vertex-colouring, 
and thus $\chiDmin(K_{m,n}) \le 1 + \big\lceil\frac{n}{m-1}\big\rceil$.

\medskip

Let us now consider distinguishing arc-labellings of $\OR{K}'$.
Observe that if the restriction of any arc-labelling $\lambda$ of $\OR{K}'$
to the subgraphs $\OR{K}'[X\cup Y_d]$ induced by $X\cup Y_d$, $1\le d\le m-1$, is distinguishing for this subgraph,
then $\lambda$ is a distinguishing arc-labelling of $\OR{K}'$.
Since for every~$d$, $1\le d\le m-1$, we have
$\DminP(\OR{K}'[X\cup Y_d]) \le \DP(K_{m,\big\lceil\frac{n}{m-1}\big\rceil})$,
we get $\DminP(K_{m,n}) \le \DP(K_{m,\big\lceil\frac{n}{m-1}\big\rceil})$.

Finally, from Proposition~\ref{prop:inequalities} and Table~\ref{table:global-edge}, we get
$$n=\chiP(K_{m,n}) \le \chiDminP(K_{m,n}) \le \chiDP(K_{m,n})=n,$$
and thus $\chiDminP(K_{m,n}) = n$,
which completes the proof of the third item.
\end{proof}

Our goal now is thus to determine for which values of $m$ and $n$, $K_{m,n}$
admits a rigid orientation. It should be noticed here that this question has been
considered in~\cite{HJ01,HR05} for mixed graphs 
(that is, graphs having both oriented and non-oriented edges), 
where rigid orientations were referred to as {\it identity orientations}.
In these two papers, the authors were interesting in
determining the smallest
number of edges of a graph which can be oriented so that the resulting
mixed graph has the trivial automorphism group.

Let us first introduce some more notation.
For a given orientation $\OR{K}$ of $K_{m,n}$,
we associate with each vertex $y_i$ from $Y$ the word $w_i=b^1_i\cdots b^m_i$
on the alphabet $\{0,1\}$, defined by $b^j_i=0$ if $x_jy_i$
is an arc, and $b^j_i=1$ otherwise. Figure~\ref{fig:orientation-K35} gives
the word associated with each vertex from $Y$ for the depicted orientation
of $K_{3,5}$.

For every integer $m\ge 2$, we will denote by $\OR{KK}^*_m$ the (unique, up to isomorphism) orientation
of the complete bipartite graph $K_{m,2^m}$ for which all the words associated with the vertices
in $Y$ are distinct, and by $(X^*,Y^*)$ the corresponding bipartition of $V(\OR{KK}^*_m)$.
This orientation will be called the \emph{canonical orientation} of $K_{m,2^m}$.
Observe that every vertex $x$ in $X^*$ has exactly $2^{m-1}$ in-neighbours and 
$2^{m-1}$ out-neighbours in $\OR{KK}^*_m$.

Let $\OR{K}$ be any orientation of $K_{m,n}$, $n\ge m\ge 2$.
We say that two vertices $u$ and $v$ are \emph{full twins} in~$\OR{K}$
if $N^+(u)=N^+(v)$ (which implies $N^-(u)=N^-(v)$).
For example, $y_1$ and $y_2$ are full twins in 
the orientation of $K_{3,5}$ depicted in Figure~\ref{fig:orientation-K35}. 
Observe that the existence of full twins in an orientation 
of a complete bipartite graph ensures that this orientation is not rigid.

\begin{proposition}\label{prop:twins-not-rigid}
Let $\OR{K}$ be any orientation of $K_{m,n}$, $n > m\ge 2$.
If there exist two full twins $u$ and $v$ in~$\OR{K}$,
then $\OR{K}$ is not rigid.
In particular, if $n>2^m$, then $K_{m,n}$ does not admit any rigid orientation.
\end{proposition}

\begin{proof}
It suffices to observe that the transposition $(u,v)$ of $V(\OR{K})$
is an automorphism of~$\OR{K}$.
The second statement follows from the fact 
that the maximum number of distinct orientations of the edges incident with
a vertex in $Y$ is $2^m$, so that every orientation of $K_{m,n}$, $n>2^m$,
necessarily contains a pair of full twin vertices.
\end{proof}

Similarly, we say that $u$ and $v$ are \emph{full antitwins} in~$\OR{K}$
if $N^+(u)=N^-(v)$ (which implies $N^-(u)=N^+(v)$).

Let now $\{x,x'\}$ be a pair of vertices from $X$.
%
We say that $y$ and $y'$ are \emph{$\{x,x'\}$-antitwins} in~$\OR{K}$ if
(i) the set of vertices $\{x,x',y,y'\}$ induces a directed 4-cycle, and
(ii) $y$ and $y'$ agree on every vertex $x''$ from $X\setminus\{x,x'\}$.
In particular,
it means that 
there is no other directed
4-cycle containing both $y$ and $y'$.
Note that any two such $\{x,x'\}$-antitwins have the same out-degree, and thus
belong to the same subset $Y_d$ of $Y$, for some integer $d$, $1\le d\le m-1$
(in particular, a source vertex or a sink vertex cannot have an $\{x,x'\}$-antitwin).
Moreover, observe that if 
$\{y,y'\}$ and $\{y,y''\}$ are both pairs of $\{x,x'\}$-antitwins in~$\OR{K}$, then
$y'$ and $y''$ are necessarily full twins in~$\OR{K}$.

For example, in the orientation of $K_{3,5}$ depicted in Figure~\ref{fig:orientation-K35},
$y_1$ and $y_2$ are full twins,
$y_1$ and $y_4$ (or $y_2$ and $y_4$) are full antitwins,
while $y_3$ and $y_5$ are $\{x_1,x_3\}$-antitwins.
Observe also that the canonical orientation $\OR{KK}^*_{m}$ of $K_{m,2^m}$
contains no pair of full twins,
that every vertex $y$ from $Y^*$ has one full antitwin,
and that every vertex $y$ from $Y^*$ which is neither a source not a sink has
an $\{x,x'\}$-antitwin for exactly $d^-(y)\times d^+(y)$ pairs of vertices $\{x,x'\}$.

\medskip

For $m=2$, we have the following result.

\begin{figure}
\begin{center}
%
%
\begin{tikzpicture}[scale=0.6]   
\node at (0,4.5) {};
\node at (0,-2.5) {};
\node[draw,circle] (X1) at (0,2) {{\footnotesize $x_1$}};
\node[draw,circle] (X2) at (0,0) {{\footnotesize $x_2$}};
\node[draw,circle] (Y1) at (5,3) {{\footnotesize $y_1$}};
\node[draw,circle] (Y2) at (5,1) {{\footnotesize $y_2$}};
\node[draw,circle] (Y3) at (5,-1) {{\footnotesize $y_3$}};
\draw[->,>=latex] (X1) to (Y1);
\draw[->,>=latex] (X2) to (Y1);
\draw[->,>=latex] (X1) to (Y2);
\draw[->,>=latex] (Y2) to (X2);
\draw[->,>=latex] (Y3) to (X1);
\draw[->,>=latex] (Y3) to (X2);
\end{tikzpicture}
\hskip 1cm
\begin{tikzpicture}[scale=0.6]   
\node at (0,4.5) {};
\node at (0,-2.5) {};
\node[draw,circle] (X1) at (0,2) {{\footnotesize $x_1$}};
\node[draw,circle] (X2) at (0,0) {{\footnotesize $x_2$}};
\node[draw,circle] (Y1) at (5,4) {{\footnotesize $y_1$}};
\node[draw,circle] (Y2) at (5,2) {{\footnotesize $y_2$}};
\node[draw,circle] (Y3) at (5,0) {{\footnotesize $y_3$}};
\node[draw,circle] (Y4) at (5,-2) {{\footnotesize $y_4$}};
\draw[->,>=latex] (X1) to (Y1);
\draw[->,>=latex] (X2) to (Y1);
\draw[->,>=latex] (X1) to (Y2);
\draw[->,>=latex] (Y2) to (X2);
\draw[->,>=latex] (Y3) to (X1);
\draw[->,>=latex] (X2) to (Y3);
\draw[->,>=latex] (Y4) to (X1);
\draw[->,>=latex] (Y4) to (X2);
\end{tikzpicture}
\caption{Orientations of $K_{2,3}$ and $K_{2,4}$ for the proof of Lemma~\ref{lem:K2n}.\label{fig:m=2}}
\end{center}
\end{figure}
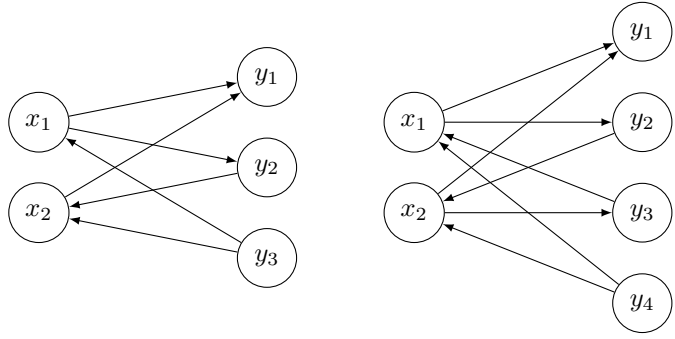

\begin{lemma}\label{lem:K2n}
The complete bipartite graph $K_{2,4}$ does not admit any rigid orientation, while
$K_{2,3}$ admits a rigid orientation.
\end{lemma}

\begin{proof}
Let $\OR{K}$ be any orientation of $K_{2,4}$. By Proposition~\ref{prop:twins-not-rigid},
we can assume that $\OR{K}$ has no full twins, that is, $\OR{K}$ is the orientation
of $K_{2,4}$ depicted in Figure~\ref{fig:m=2}. It is then easy to check that the
permutation $(x_1,x_2)(y_2,y_3)$ is an automorphism of $\OR{K}$.
A rigid orientation of $K_{2,3}$ is depicted in Figure~\ref{fig:m=2}.
\end{proof}

\medskip

We now want to characterize the values of $m$ and $n$, with $n > m\ge 2$, for which $K_{m,n}$ admits
a rigid orientation. For every such graph, by Proposition~\ref{prop:rigid-orientation},
we will then have 
$\Dmin(K_{m,n})=\DminP(K_{m,n})=1$, $\chiDmin(K_{m,n})=2$ and $\chiDminP(K_{m,n})=n$.

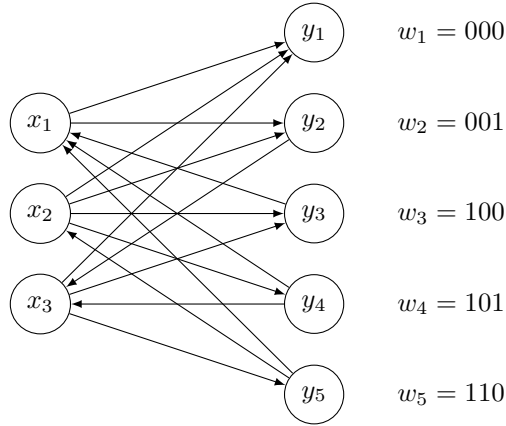
\begin{figure}
\begin{center}
\begin{tikzpicture}[scale=0.6]   
\node[draw,circle] (X1) at (0,2) {{\footnotesize $x_1$}};
\node[draw,circle] (X2) at (0,0) {{\footnotesize $x_2$}};
\node[draw,circle] (X3) at (0,-2) {{\footnotesize $x_3$}};
\node[draw,circle] (Y1) at (6,4) {{\footnotesize $y_1$}};
\node[draw,circle] (Y2) at (6,2) {{\footnotesize $y_2$}};
\node[draw,circle] (Y3) at (6,0) {{\footnotesize $y_3$}};
\node[draw,circle] (Y4) at (6,-2) {{\footnotesize $y_4$}};
\node[draw,circle] (Y5) at (6,-4) {{\footnotesize $y_5$}};
\node at (9,4) {{\footnotesize $w_1=000$}};
\node at (9,2) {{\footnotesize $w_2=001$}};
\node at (9,0) {{\footnotesize $w_3=100$}};
\node at (9,-2) {{\footnotesize $w_4=101$}};
\node at (9,-4) {{\footnotesize $w_5=110$}};
\draw[->,>=latex] (X1) to (Y1);
\draw[->,>=latex] (X2) to (Y1);
\draw[->,>=latex] (X3) to (Y1);
\draw[->,>=latex] (X1) to (Y2);
\draw[->,>=latex] (X2) to (Y2);
\draw[->,>=latex] (Y2) to (X3);
\draw[->,>=latex] (Y3) to (X1);
\draw[->,>=latex] (X2) to (Y3);
\draw[->,>=latex] (X3) to (Y3);
\draw[->,>=latex] (Y4) to (X1);
\draw[->,>=latex] (X2) to (Y4);
\draw[->,>=latex] (Y4) to (X3);
\draw[->,>=latex] (Y5) to (X1);
\draw[->,>=latex] (Y5) to (X2);
\draw[->,>=latex] (X3) to (Y5);
\end{tikzpicture}
\caption{A twin-free orientation $\protect\OR{K}$ of the complete bipartite graph $K_{3,5}$.\label{fig:K}}
\end{center}
\end{figure}

Each orientation $\OR{K}$ of a complete bipartite graph $K_{m,n}$, $n>m\ge 2$,
having no full twins
is a subdigraph of the canonical orientation $\OR{KK}^*_m$, obtained by deleting
some vertices of $Y^*$. 
In the following, such an orientation $\OR{K}$ will be described by the 
set $W(\OR{K})$ of words associated with these deleted vertices.
Note that the cardinality of the set $W(\OR{K})$ is precisely $2^m-n$.
Let $W(\OR{K}) = \{w_1,\dots,w_{2^m-n}\}$.
Note that for every vertex $x_i$ in $X$, $1\le i\le m$, we have
\begin{equation}\label{eq:in-degree}
d^-_{\OR{K}}(x_i)=2^{m-1} - \sum_{j=1}^{2^m-n} w_j^i.
\end{equation}
For example, $W(\OR{KK}^*_m)=\varnothing$,
while the set of words describing the orientation $\OR{K}$ of
$K_{3,5}$ depicted in Figure~\ref{fig:K} is
$W(\OR{K})=\{010,011,111\}$, and the in-degree of $x_3$ is $2^2-(0+1+1)=2$.

Since every orientation $\OR{K}$ of $K_{m,n}$, $3\le m < n\le 2^m$, having no full twins
is a subdigraph of $\OR{KK}^*_m$, we get that for every two distinct vertices
$y$ and $y'$ in~$Y$, there exists a vertex $x$ in~$X$ such that
$y$ and $y'$ disagree on~$x$.
Therefore, every non-trivial automorphism of $\OR{K}$ must act on~$X$.

\begin{observation}\label{obs:act-on-X}
For every orientation $\OR{K}$ of $K_{m,n}$, $3\le m < n\le 2^m$,
having no full twins, the only automorphism
of~$\OR{K}$ that fixes every vertex $x$ of~$X$ is the identity.
Therefore, if all vertices in $X$ have distinct in-degrees
(or, equivalently, distinct out-degrees), then $\OR{K}$ is rigid.
\end{observation}

Thanks to this observation, we can solve the cases $m=3$ and $m=4$.

\begin{lemma}\label{lem:K3n-K4n}
For every integer $n$, $4\le n\le 6$, $K_{3,n}$ admits a rigid orientation.
For every integer $n$, $5\le n\le 13$, $K_{4,n}$ admits a rigid orientation.
\end{lemma}

\begin{proof}
Consider the orientations $\OR{K}$, $\OR{K}'$ and $\OR{K}''$ of $K_{3,4}$, $K_{3,5}$ and $K_{3,6}$,
respectively,
given by $W(\OR{K})=\{110,010,000,111\}$,
$W(\OR{K}')=\{110,010,000\}$ and
$W(\OR{K}'')=\{110,010\}$.
In each case, according to Equation~(\ref{eq:in-degree}),
all vertices in $X$ have distinct in-degrees, which implies,
by Observation~\ref{obs:act-on-X}, that all these orientations are rigid.

The same argument applies for the orientation $\OR{K}$ of $K_{4,13}$
given by $W(\OR{K})=\{1110,1100,1000\}$, and thus 
$\OR{K}$ is a rigid orientation of $K_{4,13}$.
Now, for each $n$, $5\le n\le 12$, consider the set $W_n$ obtained
from the set $\{1110,1100,1000\}$ by adding $\left\lfloor\frac{13-n}{2}\right\rfloor$
pairs of full antitwins distinct from the pair $\{0000,1111\}$, together with the word 0000 if
$n$ is even. 
This can always be done since the set $\{0,1\}^4\setminus\{1110,1100,1000\}$ contains 
five pairs of full antitwins.
Again, in each case, we get that all vertices in $X$ have distinct in-degrees
so that, by Observation~\ref{obs:act-on-X}, all the corresponding orientations are rigid.
\end{proof}

\begin{figure}
\begin{center}
\begin{tabular}{ccccccc|l}
\footnotesize{1} & \footnotesize{2} & \footnotesize{3} & \footnotesize{4} & 
\footnotesize{5} & \footnotesize{6} & \footnotesize{7} & $j$ \\
\hline
0 & 0 & 1 & 1 & 1 & 1 & 1 \\
0 & 1 & 0 & 0 & 1 & 1 & 1 \\
0 & 0 & 0 & 1 & 0 & 0 & 1 & The set $W_7$\\
0 & 0 & 0 & 0 & 0 & 1 & 0 \\
\hline
\hline
0 & 1 & 1 & 2 & 2 & 3 & 3 & The sums $S^j(W_7)$\\
\end{tabular}

\vskip 1cm

\begin{tabular}{cccccccc|l}
\footnotesize{1} & \footnotesize{2} & \footnotesize{3} & \footnotesize{4} & 
\footnotesize{5} & \footnotesize{6} & \footnotesize{7} & \footnotesize{8} & $j$ \\
\hline
0 & 0 & 1 & 1 & 1 & 1 & 1 & 1\\
0 & 1 & 0 & 0 & 1 & 1 & 1 & 1\\
0 & 0 & 0 & 1 & 0 & 0 & 1 & 1 & The set $W_8$\\
0 & 0 & 0 & 0 & 0 & 1 & 0 & 1\\
\hline
\hline
0 & 1 & 1 & 2 & 2 & 3 & 3 & 4 & The sums $S^j(W_8)$\\
\end{tabular}
\caption{The sets $W_7$ and $W_8$ for the proof of Lemma~\ref{lem:m-over-2-to-2m-m}.\label{fig:Wstar}}
\end{center}
\end{figure}

We are now able to prove the following result.

\begin{lemma}\label{lem:m-over-2-to-2m-m}
For every two integers $m$ and $k$, $m\ge 5$, $\left\lceil\frac{m}{2}\right\rceil\le k < 2^m-m$,
$K_{m,2^m-k}$ admits a rigid orientation.
\end{lemma}

\begin{proof}
For every set of words $W=\{w_1,\dots,w_k\}\subseteq \{0,1\}^m$, we denote
by $S^i(W)$, $1\le i\le m$, the sum of the $i$-th symbols of the words in $W$,
that is, $S^i(W) = \sum_{\ell=1}^k w_\ell^i$.

We first consider the case $k=\left\lceil\frac{m}{2}\right\rceil$.
We will construct a particular set $W_m=\{w_1,\dots,w_k\}$ of $k$ words such
that the orientation $\OR{K}$ of $K_{m,2^m-\left\lceil\frac{m}{2}\right\rceil}$ 
with $W(\OR{K})=W_m$ is rigid.
These $k$ words are defined as follows (see Figure~\ref{fig:Wstar} for the 
sample sets $W_7$ and $W_8$):

\begin{itemize}
\item $w_1 = 001^{m-2}$,
\item for every $i$, $2\le i\le \left\lceil\frac{m}{2}\right\rceil - 1$,
$w_i = 0^{2i-3}1001^{m-2i}$,
\item $w_{\left\lceil\frac{m}{2}\right\rceil} = 0^{m-2}10$ if
$m$ is odd, or $w_{\left\lceil\frac{m}{2}\right\rceil} = 0^{m-3}101$ if $m$ is even.
\end{itemize}

Observe that $S^1(W_m)=0$,
$S^{2i}(W_m)=S^{2i+1}(W_m)=i$ for every $i$, $1\le i\le \left\lfloor\frac{m-1}{2}\right\rfloor$,
and $S^m(W_m)=\frac{m}{2}$ if $m$ is even.
By Equation~(\ref{eq:in-degree}), this implies that the set of pairs of vertices in~$X$
having the same in-degree is
$\left\{(x_{2i},x_{2i+1}):\ 1\le i\le \left\lfloor\frac{m-1}{2}\right\rfloor\right\}$. 
Therefore, the restriction to $X$ of every non-trivial automorphism of $\OR{K}$ must be the product
of at least one transposition corresponding to these pairs of vertices.

We claim that such a situation cannot occur, so that $\OR{K}$ does not admit
any non-trivial automorphism.
%
Indeed, we cannot exchange $x_{2i}$ and $x_{2i+1}$, $1\le i\le  \left\lfloor\frac{m-1}{2}\right\rfloor$,
since the pair $\{y,y'\}$ of $\{x_{2i},x_{2i+1}\}$-antitwins,
where $y$ is the vertex associated with the word
$w_{i}$, 
is such that $y\in W_m$, and thus $y$ has been deleted, 
while $y'\notin W_m$, and thus $y'$ still belongs to $\OR{K}$
(hence, $x_{2i}$ and $x_{2i+1}$ disagree on $y'$).
The orientation $\OR{K}$ of $K_{m,2^m-k}$ given by $W(\OR{K})=W_m$ is thus rigid.

\medskip

Let us now consider the case $\left\lceil\frac{m}{2}\right\rceil < k < 2^m - m$,
and let $p = k - \left\lceil\frac{m}{2}\right\rceil$.
We then construct a set of words $W_k$, by adding to the set
$W_m$ previously defined $\left\lfloor\frac{p}{2}\right\rfloor$ pairs
of full antitwins distinct from the pair $\{0^m,1^m\}$, together with the word $0^m$
if $p$ is odd.
Note that this is always possible since 
$W_m$ contains $\left\lceil\frac{m}{2}\right\rceil$
words, and thus $\frac{1}{2}\left(2^m - 2\left\lceil\frac{m}{2}\right\rceil\right)$ 
pairs of antitwins are available while we need at most
$\frac{1}{2}\left(2^m - m - \left\lceil\frac{m}{2}\right\rceil\right)$ such pairs.
Doing so, the sums $S^j(W_k)$, $1\le j\le m$, satisfy the same
properties as in the previous case so that, again,
the set of pairs of vertices in $X$
having the same in-degree is
$\left\{(x_{2i},x_{2i+1}):\ 1\le i\le \left\lfloor\frac{m-1}{2}\right\rfloor\right\}$. 
The same argument as before then allows us to conclude that the orientation
$\OR{K}$ of $K_{m,2^m-k}$ given by $W(\OR{K})=W_k$ is also rigid.
This completes the proof.
\end{proof}

\medskip

In order to complete the study of complete bipartite graphs, we still
have to consider the graphs $K_{m,2^m-k}$ for every $k$, $0\le k<m$.
The following result is useful for small values of $k$.

\begin{lemma}\label{lem:2-colonnes-identiques}
Let $\OR{K}$ be any orientation of $K_{m,n}$, $3\le m < n \le 2^m$, not containing full twins.
If there exists two indices $i$ and $i'$, $1\le i<i'\le m$, such that
$w_j^i = w_j^{i'}$ for every word $w_j\in W(\OR{K})$, then $\OR{K}$ is not rigid.
In particular, if $m\ge 3$ and $0\le k<\log_2(m)$, 
then $K_{m,2^m-k}$ does not admit any rigid orientation.
\end{lemma}

\begin{proof}
Observe that there are exactly $p=2^{m-2}$ pairs of $\{x_i,x_{i'}\}$-antitwins in $Y$ in
the canonical orientation $\OR{KK}^*_m$ of $K_{m,n}$.
Let us denote the set
of these pairs by $A_{x_i,x_{i'}}=\{(y_{j_1},y'_{j_1}),\dots,(y_{j_p},y'_{j_p})\}$.
Since all vertices corresponding to words in $W(\OR{K})$ agree on $x_i$ and $x_{i'}$,
all vertices belonging to a pair of $A_{x_i,x_{i'}}$ belong to $\OR{K}$.
Therefore, the permutation $(x_i,x_{i'})(y_{j_1},y'_{j_1})\cdots(y_{j_p},y'_{j_p})$ 
is a non-trivial automorphism of $\OR{K}$, and thus $\OR{K}$ is not rigid.

Finally, if $m\ge 3$ and $0\le k<\log_2(m)$, 
then, for every orientation  $\OR{K}'$ of $K_{m,2^m-k}$,
there necessarily exist two indices $i$ and $i'$, $1\le i<i'\le m$, such that
$w_j^i = w_j^{i'}$ for every word $w_j\in W(\OR{K}')$, 
which implies that $\OR{K}'$ is not rigid.
\end{proof}

From Theorem~\ref{th:complete-bipartite-graphs}(2) and 
Lemmas \ref{lem:K2n}, \ref{lem:K3n-K4n} and~\ref{lem:m-over-2-to-2m-m}, 
we finally get the following corollary.

\begin{corollary}\label{cor:final-complete-bipartite}
For every two integers $m$ and $n$, $m<n$, with either 
$m\ge 5$ and $n\le 2^m-\left\lceil\frac{m}{2}\right\rceil$, or
$(m,n) \in \{(2,3), (3,4), (3,5), (3,6)\}\cup\{(4,p),\ 5\le p\le 13\}$, we have
$$\Dmin(K_{m,n}) = \DminP(K_{m,n}) = 1,\ \chiDmin(K_{m,n}) = 2\ \mbox{and}\ \chiDminP(K_{m,n}) = n.$$
\end{corollary}

\section{Discussion}
\label{sec:discussion}

In this paper, we have studied the 
distinguishing number, the distinguishing index, 
the distinguishing chromatic number and the distinguishing chromatic index
of oriented graphs.

We have determined the minimum and maximum values, taken over all possible
orientations of the corresponding underlying graph, of these parameters
for paths, cycles, complete graphs and bipartite complete graphs,
except for the minimum values for unbalanced bipartite complete graphs
$K_{m,n}$ not covered by Corollary~\ref{cor:final-complete-bipartite},
in which case we were only able to provide upper bounds (see Theorem~\ref{th:complete-bipartite-graphs}(3)).

\medskip

Following our work, and apart the question of considering other graph classes,
the main question is thus to determine the minimum values of 
the distinguishing parameters of unbalanced bipartite complete graphs
not covered by Corollary~\ref{cor:final-complete-bipartite}, that is,
of $K_{m,n}$ with $m\ge 5$ and  
$n > 2^m-\left\lceil\frac{m}{2}\right\rceil$.
In particular, it would be interesting to know which of those complete bipartite
graphs admit a rigid orientation.
By Lemma~\ref{lem:2-colonnes-identiques}, we know that
$K_{m,n}$, $m<n$, does not admit any rigid orientation when $n > 2^m - \log_2(m)$.
We can also prove that $K_{m,n}$ does not admit any rigid orientation when $n= 2^m - 2^p$,
for some $p\ge 2$.
However, we do not have any complete characterization of these graphs yet.


\end{document}